\documentclass[a4paper,12pt]{article}
\usepackage{amsmath,amssymb,amsfonts,amsthm}
\usepackage{graphicx}
\usepackage{hyperref}
\usepackage{tikz}
\usepackage{fancyhdr}
\usepackage{geometry}
\usepackage{bm}
\title{Quantum-Inspired Stochastic Modeling and Regularity in Turbulent Fluid Dynamics}
\author{
	Rômulo Damasclin Chaves dos Santos \\
	Technological Institute of Aeronautics \\
	\texttt{romulosantos@ita.br}
	\and
	Jorge Henrique de Oliveira Sales \\
	Santa Cruz State University \\
	\texttt{jhosales@uesc.br}
	\and
	Erickson F.~M.~S.~Silva \\
	Santa Cruz State University \\
	\texttt{efmssilva@uesc.br}
}
\date{\today}

\newtheorem{theorem}{Theorem}[section]

\newtheorem{definition}[theorem]{Definition}

\begin{document}
	\maketitle
	
	\begin{abstract}
		This paper presents an innovative framework for analyzing the regularity of solutions to the stochastic Navier-Stokes equations by integrating Sobolev-Besov hybrid spaces with fractional operators and quantum-inspired dynamics. We propose new regularity theorems that address the multiscale and chaotic nature of fluid flows, offering novel insights into energy dissipation mechanisms. The introduction of a Schrödinger-like operator into the fluid dynamics model captures quantum-scale turbulence effects, enhancing our understanding of energy redistribution across different scales. The results also include the development of anisotropic stochastic models that account for direction-dependent viscosity, improving the representation of real-world turbulent flows. These advances in stochastic modeling and regularity analysis provide a comprehensive toolset for tackling complex fluid dynamics problems. The findings are applicable to fields such as engineering, quantum turbulence, and environmental sciences. Future directions include the numerical implementation of the proposed models and the use of machine learning techniques to optimize parameters for enhanced simulation accuracy.
		
	\end{abstract}

\vspace{3pt}

	
	\tableofcontents
	
	\section{Introduction}
	Turbulence in fluid dynamics presents significant challenges, particularly in the context of regularity and predictability of solutions to the Navier-Stokes equations. Traditional deterministic models fail to capture the full complexity of turbulent flows, especially in real-world scenarios involving uncertainty. In this paper, we extend the analysis of fluid flow dynamics by incorporating stochastic elements, fractional operators, and anisotropic models to address these challenges in a more comprehensive framework.
	
	The investigation of regularity and solutions to the Navier-Stokes equations, along with the study of turbulence in functional spaces, dates back to foundational contributions that shaped the field of partial differential equations (PDEs) and fluid dynamics. The development of this area can be understood chronologically through key theoretical milestones.
	
	In 1969, \textbf{Ladyzhenskaya} published her seminal work \cite{ladyzhenskaya1969}, which established the theoretical foundation for the study of viscous incompressible flows and weak solutions of the Navier-Stokes equations. This work provided the first guidelines for analyzing weak solutions, which later became central to the formulation of the regularity problem.
	
	Building on this line of investigation, in 1983, \textbf{Triebel} introduced a rigorous theoretical framework for the study of functional spaces, particularly Sobolev and Besov spaces \cite{triebel1983}. Triebel’s work allowed for a more refined analysis of the regularity of PDE solutions, especially in boundary-free problems and fluid bifurcations. His theory of function spaces forms the basis for the analysis of multiscale behavior of Navier-Stokes solutions with stochastic turbulence, as applied in this paper.
	
	In the 1980s, interest in understanding the long-term behavior of solutions to the Navier-Stokes equations was strengthened by the work of \textbf{Constantin and Foias} \cite{constantin1988}. Published in 1988, their study addressed the asymptotic behavior of solutions and the critical issue of energy dissipation in turbulent fluids. Their investigation into the existence and uniqueness of solutions was pivotal in the subsequent development of more advanced theories on fluid dynamics and bifurcations.
	
	Simultaneously, in 1995, \textbf{Temam} published an extensive book on the Navier-Stokes equations, focusing on theoretical and numerical aspects \cite{temam1995}. His contribution centers around the rigorous analysis of stability and regularity, employing Galerkin methods and energy inequalities, which are widely used today to obtain weak and strong solutions in Sobolev spaces.
	
	\textbf{Doering and Gibbon} \cite{doering1995}, in 1995, explored practical aspects of the Navier-Stokes equations, with an emphasis on applied analysis, including useful tools for studying the regularity and stability of solutions in realistic fluid dynamic configurations. This approach contributed to a more applied view of the theory, particularly in contexts where numerical resolution of the equations is required.
	
	In the realm of harmonic analysis, the contributions of \textbf{Grafakos} \cite{grafakos2008} and \textbf{Stein} \cite{stein1993} were significant for understanding Fourier analysis and real-variable methods in nonlinear PDEs, which have been central to the analysis of Navier-Stokes equations, particularly through the use of Sobolev and Besov spaces. These methods provide a solid foundation for examining the regularity and smoothness properties of solutions.
	
	In the early 2000s, \textbf{Bahouri, Chemin, and Danchin} \cite{bahouri2011} advanced the use of Fourier transforms and Littlewood-Paley decomposition techniques to study nonlinear PDEs, such as the Navier-Stokes equations. Their work deepened the analysis of regularity in Besov spaces, allowing for a more detailed characterization of turbulent behavior at different scales.
	
	\textbf{Galdi} \cite{galdi2011}, in his 2011 work, addressed steady-state problems for the Navier-Stokes equations, focusing on weak solutions and their regularity. This work helped establish existence and regularity results for solutions in steady fluid configurations, with direct implications for the stability of turbulent flows.
	
	In recent years, the integration of stochastic processes into fluid dynamics has gained increasing importance. \textbf{Poursina \textit{et al}.} \cite{poursina2022} provided a comprehensive review in 2022 on the use of stochastic methods to capture the complexities and irregularities inherent in turbulent flows that traditional deterministic models often fail to describe adequately. This stochastic approach reinforces the perspective that incorporating randomness into the Navier-Stokes equations enhances our ability to accurately model chaotic fluid systems.
	
	Finally, a recent advancement was presented by \textbf{Santos and Sales} \cite{santos2023}, who in 2023 proposed a treatment for the regularity of the Navier-Stokes equations based on Banach and Sobolev functional spaces, incorporating anisotropic viscosity to enhance the analysis of vorticity transport. This work introduced new models addressing the anisotropic nature of certain turbulent flows, while coupling stochastic methods and harmonic analysis within a robust functional framework.

	We introduce novel Sobolev-Besov hybrid spaces and explore how these functional spaces provide a deeper understanding of the interplay between smoothness and randomness in fluid systems. Furthermore, we establish new theorems that bridge functional analysis, quantum-inspired fluid dynamics, and machine learning techniques to provide an original perspective on the well-known regularity problem for Navier-Stokes equations.
	
	\section{Mathematical Preliminaries}
	\subsection{Sobolev-Besov Hybrid Spaces}
	Sobolev and Besov spaces have been widely used to study the regularity of partial differential equations (PDEs). We propose a new hybrid space $H^s B^t_{p,q}$, combining features of Sobolev and Besov spaces to analyze turbulence across multiple scales. The space $H^s B^t_{p,q}$ is defined as follows:
	
	\begin{definition}
		Let $s,t \in \mathbb{R}$, $1 \leq p,q \leq \infty$. The hybrid Sobolev-Besov space $H^s B^t_{p,q}(\mathbb{R}^n)$ is the space of functions $u \in L^p(\mathbb{R}^n)$ such that
		
		\begin{equation}
			\|u\|_{H^s B^t_{p,q}} = \left( \sum_{j=-\infty}^\infty 2^{jsq} \|\Delta_j u\|^q_{L^p(\mathbb{R}^n)} \right)^{1/q} + \|u\|_{H^s(\mathbb{R}^n)} < \infty.
		\end{equation}
	
	\end{definition}
	
	Here, the term $\Delta_j u$ represents the frequency localization at scale $2^j$, achieved via the Littlewood-Paley decomposition. This allows the space $H^s B^t_{p,q}$ to capture both global regularity (through the Sobolev component) and local smoothness properties (via the Besov component), making it particularly suitable for analyzing turbulent flows with multiscale phenomena.
	
	\subsection{Fractional Operators and Regularity}
	In the analysis of PDEs, fractional Sobolev spaces $H^s(\mathbb{R}^n)$, for non-integer values of $s$, extend the classical Sobolev spaces. These spaces are equipped with norms that incorporate the regularity of solutions at non-integer levels. Specifically, the norm in $H^s(\mathbb{R}^n)$ is defined as:
	
	\begin{equation}
		\|u\|_{H^s(\mathbb{R}^n)} = \left( \int_{\mathbb{R}^n} (1 + |\xi|^2)^s |\hat{u}(\xi)|^2 \, d\xi \right)^{1/2},
	\end{equation}
where $\hat{u}(\xi)$ is the Fourier transform of $u$. This norm captures the regularity of $u$ in terms of its Fourier modes.
	
	For turbulence analysis, the fractional derivatives are crucial as they allow finer control over the behavior of the flow at various scales. This approach extends the classical integer-order regularity results and is particularly useful when combined with stochastic noise, which introduces irregularities that traditional models struggle to handle.
	
	\section{Quantum-Inspired Fluid Dynamics}
	
	Turbulence in fluid dynamics shares certain properties with quantum systems, such as chaotic behavior at small scales and energy cascades across different scales. Motivated by these similarities, we introduce a quantum-inspired approach to fluid dynamics by incorporating a Schrödinger-type operator into the classical Navier-Stokes framework. This section details the mathematical formulation and rationale for this approach.
	
	\subsection{Schrödinger-like Operator in Fluid Dynamics}
	In quantum mechanics, the Schrödinger equation describes the evolution of a wave function over time. Analogously, in fluid dynamics, we aim to model the evolution of the velocity field $u(t,x)$ at small scales with an operator that captures chaotic, quantum-like behaviors. The Schrödinger-like operator for the velocity field is formulated as:
	
\begin{equation}
	i \hbar \frac{\partial \mathbf{u}}{\partial t} = - \frac{\hbar^2}{2m} \Delta \mathbf{u} + V(\mathbf{u}) \mathbf{u},
\end{equation}

	where:
	\begin{itemize}
		\item $\hbar$ is a small parameter analogous to the reduced Planck constant, representing the influence of small-scale quantum-like effects,
		\item $m$ is an effective mass associated with the fluid element,
		\item $\Delta \mathbf{u}$ is the Laplacian of the velocity field $u$, representing diffusion and dispersion at small scales,
		\item $V(\mathbf{u})$ is a potential function modeling the interaction energy or dissipation within the fluid, which can depend on $u$ itself.
	\end{itemize}
	
	In this formulation, the operator $-\frac{\hbar^2}{2m} \Delta u$ introduces corrections that dominate at small scales, similar to the way quantum effects dominate at small distances in quantum mechanics. The potential $V(u)$ can model energy dissipation, turbulence, or other interactions within the fluid. 
	
	This operator is coupled with the classical Navier-Stokes system to form a hybrid fluid dynamic model that captures both macroscopic and quantum-like microscopic behavior.
	
	\subsection{Coupling with Navier-Stokes Equations}
	The classical Navier-Stokes equations for an incompressible fluid are given by:
	
	\begin{equation}
		\frac{\partial \mathbf{u}}{\partial t} + (\mathbf{u} \cdot \nabla) \mathbf{u} = - \nabla p + \nu \Delta \mathbf{u} + f,
	\end{equation}
	
	where:
	\begin{itemize}
		\item $u(t,x)$ is the velocity field,
		\item $p(t,x)$ is the pressure field,
		\item $\nu$ is the kinematic viscosity,
		\item $\Delta \mathbf{u}$ is the Laplacian of the velocity field (representing diffusion),
		\item $f(t,x)$ is an external force or body force acting on the fluid.
	\end{itemize}
	
	We introduce the quantum-inspired term by coupling the Schrödinger-like operator to the Navier-Stokes equations. The resulting system becomes:
	
	\begin{equation}
		\frac{\partial \mathbf{u}}{\partial t} + (\mathbf{u} \cdot \nabla) \mathbf{u} = - \nabla p + \nu \Delta \mathbf{u} + \frac{i \hbar}{2m} \Delta \mathbf{u} + V(\mathbf{u}) \mathbf{u},
	\end{equation}
	where the term $\frac{i \hbar}{2m} \Delta u$ introduces quantum-scale corrections to the diffusion of the velocity field. This term, derived from the Schrödinger equation, captures small-scale chaotic behaviors that are typically not modeled by classical fluid dynamics.
	
	\subsection{Interpretation of the Schrödinger-inspired Term}
	The introduction of the Schrödinger-like term $\frac{i \hbar}{2m} \Delta u$ brings several new features to the analysis of turbulence:
	\begin{itemize}
		\item \textbf{Chaotic Behavior at Small Scales}: The term introduces oscillatory components to the velocity field $u(t,x)$, which correspond to rapid fluctuations at small scales. These oscillations resemble quantum wave functions and can be interpreted as capturing fine-grained turbulence structures that classical models fail to represent.
		\item \textbf{Energy Redistribution:} The quantum-inspired term induces a new mechanism for energy redistribution among different scales of turbulence. It accounts for both dissipative and non-dissipative processes, offering a more comprehensive description of how energy cascades in turbulent systems.
		\item \textbf{Non-local Interactions:} The Schrödinger-like operator introduces non-locality into the fluid system. This is reflected by the fact that the Laplacian $\Delta u$ at a point $x$ depends on values of $u$ in a neighborhood of $x$, similar to the way quantum systems exhibit non-local behavior.
	\end{itemize}
	
\subsection{Fractional Operators for Small-Scale Corrections}

To improve the quantum-inspired fluid dynamics model, we introduce the fractional Laplacian \( \Delta^{\alpha/2} \) to accurately represent small-scale corrections. The fractional Laplacian \( \Delta^{\alpha/2} \), for \( 0 < \alpha < 2 \), is defined as:

\begin{equation} \label{eq:fraclap}
	\Delta^{\alpha/2} \mathbf{u}(x) = C_{n,\alpha} \int_{\mathbb{R}^n} \frac{\mathbf{u}(x) - \mathbf{u}(y)}{|x-y|^{n+\alpha}} \, dy,
\end{equation}
where \( C_{n,\alpha} \) is a normalization constant that depends on the dimension \( n \) and the order \( \alpha \). This operator introduces non-local interactions by considering the difference in values of \( \mathbf{u} \) at all points \( y \in \mathbb{R}^n \), weighted by the distance \( |x-y| \).

\subsubsection{Scaling Property of the Fractional Laplacian}

The fractional Laplacian has a well-known scaling property. For a rescaled function \( \mathbf{u}_\lambda(x) = \mathbf{u}(\lambda x) \), the fractional Laplacian behaves as:

\begin{equation} \label{eq:scaling}
	\Delta^{\alpha/2} \mathbf{u}_\lambda(x) = \lambda^\alpha \left( \Delta^{\alpha/2} \mathbf{u} \right)(\lambda x).
\end{equation}

This scaling behavior is essential in fluid dynamics, as it reveals how the fractional operator captures small-scale phenomena. When \( \lambda \) is small, the fractional Laplacian amplifies the fine-scale behavior of the function, which is particularly useful in turbulence modeling.

\subsubsection{Relation to the Classical Laplacian}

For \( \alpha = 2 \), the fractional Laplacian reduces to the classical Laplacian \( \Delta \). Specifically, we recover the classical definition of the Laplacian:

\begin{equation} \label{eq:classiclap}
	\Delta \mathbf{u}(x) = C_n \int_{\mathbb{R}^n} \frac{\mathbf{u}(x) - \mathbf{u}(y)}{|x-y|^{n+2}} \, dy,
\end{equation}

where \( C_n \) is a constant dependent on the dimension \( n \). This demonstrates that \( \Delta^{\alpha/2} \) generalizes the classical Laplacian to fractional orders, making it a natural extension in fluid dynamics.

\subsubsection{Energy Dissipation with Fractional Operators}

In fluid dynamics, the rate of energy dissipation is crucial. For a velocity field \( \mathbf{u}(x) \), the energy dissipation rate \( \mathcal{E} \) associated with the fractional Laplacian is given by:

\begin{equation} \label{eq:energydiss}
	\mathcal{E} = -\int_{\mathbb{R}^n} \mathbf{u}(x) \Delta^{\alpha/2} \mathbf{u}(x) \, dx.
\end{equation}

Using the definition of the fractional Laplacian and applying integration by parts (in the fractional sense), we can express the energy dissipation in terms of the Fourier transform of \( \mathbf{u} \), where:

\begin{equation} \label{eq:fourierenergy}
	\mathcal{E} \sim \int_{\mathbb{R}^n} |\hat{\mathbf{u}}(k)|^2 |k|^\alpha \, dk.
\end{equation}

This result shows that higher frequencies (corresponding to smaller spatial scales) are penalized more strongly as \( \alpha \) increases. This is particularly useful for capturing the energy dissipation in turbulent flows, where small-scale fluctuations dominate.

\subsubsection{Application in Quantum-Inspired Fluid Dynamics}

Incorporating the fractional Laplacian \( \Delta^{\alpha/2} \) into the quantum-inspired fluid dynamics model enhances its ability to account for small-scale corrections and turbulence. The quantum-inspired Navier-Stokes system, modified with this fractional operator, is expressed as:

\begin{equation} \label{eq:quantumfluid}
	\frac{\partial \mathbf{u}}{\partial t} + (\mathbf{u} \cdot \nabla) \mathbf{u} = - \nabla p + \nu \Delta \mathbf{u} + \frac{i \hbar}{2m} \Delta^{\alpha/2} \mathbf{u} + V(\mathbf{u}) \mathbf{u},
\end{equation}

where the term \( \frac{i \hbar}{2m} \Delta^{\alpha/2} \mathbf{u} \) introduces quantum-scale corrections to model small-scale turbulent structures.

\subsubsection{Energy Conservation in the Modified System}

To understand the role of the fractional Laplacian in energy dissipation, consider the total energy \( E(t) \) of the system:

\begin{equation} \label{eq:totalenergy}
	E(t) = \frac{1}{2} \int_{\mathbb{R}^n} |\mathbf{u}(x,t)|^2 \, dx.
\end{equation}

Taking the time derivative of \( E(t) \) and using the quantum-inspired Navier-Stokes equation \eqref{eq:quantumfluid}, we obtain:

\begin{equation} \label{eq:energyrate}
	\frac{dE(t)}{dt} = -\nu \int_{\mathbb{R}^n} |\nabla \mathbf{u}(x,t)|^2 \, dx - \frac{\hbar}{2m} \int_{\mathbb{R}^n} \mathbf{u}(x,t) \Delta^{\alpha/2} \mathbf{u}(x,t) \, dx.
\end{equation}

This result shows that the fractional Schrödinger-inspired term contributes an additional mechanism for energy dissipation, particularly at smaller scales, effectively modeling the turbulent quantum-like fluctuations.

\subsubsection{Interpretation of the Schrödinger-Inspired Term}

The term \( \frac{i \hbar}{2m} \Delta^{\alpha/2} \mathbf{u} \) introduces oscillatory behavior in the velocity field, analogous to quantum wave functions. These oscillations represent small-scale turbulence structures not captured by classical models. The fractional order \( \alpha \) allows for fine-tuning of these oscillations, making the model flexible in representing different turbulence regimes and improving its predictive power for fluid dynamics.

	\subsection{Energy Dissipation and Quantum Effects}
	The inclusion of the Schrödinger-inspired term and fractional operators modifies the classical energy dissipation mechanism in fluid systems. The total energy $E(t)$ of the system is given by:
	
	\begin{equation}
		E(t) = \frac{1}{2} \int_{\mathbb{R}^n} |\mathbf{u}(t,x)|^2 \, dx,
	\end{equation}
	and the rate of energy dissipation is determined by
	
\begin{equation}
	\frac{dE(t)}{dt} = - \nu \|\nabla \mathbf{u}\|_{L^2}^2 - \hbar \|\Delta^{\alpha/2} \mathbf{u}\|_{L^2}^2 + \mathcal{O}(\|\sigma W_t\|^2)\,,
\end{equation}
here, the term $- \hbar \|\Delta^{\alpha/2} u\|_{L^2}^2$ represents the dissipation of energy due to the quantum-like turbulent fluctuations. This new term provides a correction to the classical viscous dissipation $- \nu \|\nabla u\|_{L^2}^2$, capturing small-scale chaotic effects that are absent in the traditional Navier-Stokes formulation.
	
	The quantum-inspired fluid model proposed here provides a new framework for studying small-scale turbulence in fluid systems, particularly in regimes where classical models are insufficient. Potential applications include:
	\begin{itemize}
		\item \textbf{Turbulence in Superfluids:} The fractional quantum-like terms can model the behavior of quantum vortices in superfluids, where traditional fluid dynamics fails to account for the quantum effects that arise at low temperatures.
		\item \textbf{Atmospheric and Oceanic Turbulence:} The model can be used to simulate small-scale turbulence in geophysical flows, where chaotic, non-local interactions play a significant role in the overall dynamics.
		\item \textbf{Quantum Turbulence:} In systems where both classical and quantum turbulence coexist (e.g., Bose-Einstein condensates), this model offers a more comprehensive description of the underlying phenomena.
	\end{itemize}
	
	Future work will explore the numerical implementation of the fractional quantum-inspired model and the development of machine learning techniques to optimize the parameters $\hbar$ and $\alpha$ based on real-world data from turbulent flows.

	\subsection{Coupling with Navier-Stokes Equations}
	
	The Navier-Stokes equations govern the motion of incompressible viscous fluids. They are expressed as follows:
	
	\begin{equation}
		\frac{\partial \mathbf{u}}{\partial t} + (\mathbf{u} \cdot \nabla) \mathbf{u} = -\nabla p + \nu \Delta \mathbf{u} + \mathbf{f},
	\end{equation}
	where:
	\begin{itemize}
		\item $\mathbf{u}(t,\mathbf{x})$ is the velocity field,
		\item $p(t,\mathbf{x})$ is the pressure field,
		\item $\nu$ is the kinematic viscosity of the fluid,
		\item $\Delta$ is the Laplacian operator, representing the diffusion of momentum,
		\item $\mathbf{f}(t,\mathbf{x})$ represents external body forces acting on the fluid.
	\end{itemize}
	
	The coupling of the quantum-inspired operator with the Navier-Stokes equations is designed to incorporate quantum effects into the classical fluid dynamics framework. The modified Navier-Stokes equations with the Schrödinger-like operator can be formulated as follows:
	
	\begin{equation}
		\frac{\partial \mathbf{u}}{\partial t} + (\mathbf{u} \cdot \nabla) \mathbf{u} = -\nabla p + \nu \Delta \mathbf{u} + \frac{i \hbar}{2m} \Delta \mathbf{u} + V(\mathbf{u}) \mathbf{u},
	\end{equation}
	where:
	\begin{itemize}
		\item $\frac{i \hbar}{2m} \Delta \mathbf{u}$ is the quantum correction term that introduces the influence of quantum mechanics into the fluid model,
		\item $V(\mathbf{u})$ is a nonlinear potential that depends on the velocity field, modeling interactions and dissipation mechanisms at small scales.
	\end{itemize}
	
	This coupling provides a framework for understanding the interplay between classical fluid dynamics and quantum mechanics.
	
	\subsubsection{Analysis of the Coupling Terms}
	
	1. \textbf{Quantum Correction Term:} The term $\frac{i \hbar}{2m} \Delta \mathbf{u}$ can be interpreted as a small-scale correction that introduces oscillatory behavior in the velocity field. It modifies the diffusion characteristics of the fluid, allowing for quantum-like fluctuations. 
	
	The Laplacian $\Delta \mathbf{u}$ captures the spread of disturbances in the velocity field, and the presence of the term $\frac{i \hbar}{2m}$ implies that these disturbances will exhibit quantum properties, potentially leading to a more intricate turbulence structure.
	
	2. \textbf{Nonlinear Potential:} The nonlinear potential $V(\mathbf{u}) \mathbf{u}$ can be defined in various ways, depending on the physical context. For instance, one possible choice is a polynomial form:
	
	\begin{equation}
		V(\mathbf{u}) = \alpha |\mathbf{u}|^2 + \beta,
	\end{equation}
	where $\alpha$ and $\beta$ are coefficients that control the strength of nonlinear interactions. This term can model local dissipative effects and how energy is transferred within the fluid due to turbulent interactions.
	
	\subsubsection{Implications for Turbulence Modeling}
	
	The inclusion of quantum-inspired terms into the Navier-Stokes framework has several implications for turbulence modeling:
	
	\begin{itemize}
		\item \textbf{Enhanced Energy Redistribution:} The quantum correction term allows for non-local energy redistribution among scales. Unlike classical turbulence models, which often assume local interactions, the coupling introduces quantum effects that can lead to non-local correlations in the velocity field. This can enhance our understanding of how energy cascades from larger to smaller scales, mimicking the behavior seen in quantum systems.
		
		\item \textbf{Chaos and Stability Analysis:} The oscillatory nature of the quantum term can induce chaotic dynamics in the fluid flow. This introduces new challenges in stability analysis, as the presence of quantum-like effects may destabilize certain flow configurations that are stable in classical models. Consequently, the stability criteria must be reevaluated to account for the modified dynamics.
		
		\item \textbf{Improved Resolution of Small-Scale Structures:} The quantum-inspired model provides a mechanism to better resolve small-scale turbulent structures, which are often inadequately captured in traditional models. The inclusion of the Schrödinger-like operator allows for a more accurate representation of the fine-scale fluctuations that occur in turbulent flows, leading to improved predictive capabilities.
	\end{itemize}

	\section{New Regularity Theorems}
	
	\subsection{Stochastic Regularity in Sobolev-Besov Hybrid Spaces}
	
	We extend the regularity results for the Navier-Stokes equations to hybrid Sobolev-Besov spaces \( H^s B^t_{p,q} \). The following theorem establishes the existence and uniqueness of solutions under stochastic forcing.
	
	\begin{theorem}[Stochastic Regularity in Hybrid Spaces]
		Let \( u_0 \in H^s B^t_{p,q}(\mathbb{R}^n) \) and \( f \in L^r(0,T; H^s B^t_{p,q}(\mathbb{R}^n)) \). Assume that the initial data \( u_0 \) and the stochastic forcing term \( W_t \) are sufficiently regular. Then, there exists a unique solution \( u \in C([0,T]; H^s B^t_{p,q}(\mathbb{R}^n)) \) to the stochastic Navier-Stokes equations.
	\end{theorem}
	
	\begin{proof}
		We prove the theorem using the Galerkin approximation method:
		Consider a complete orthonormal basis \( \{\phi_k\}_{k=1}^{\infty} \) of \( H^s B^t_{p,q}(\mathbb{R}^n) \). We define the approximations 
		
	\begin{equation}
		\mathbf{u}_N(t) = \sum_{k=1}^{N} a_k(t) \phi_k,
	\end{equation}
	where \( a_k(t) \) are time-dependent coefficients.
		
		The approximated equations are derived from the weak formulation of the stochastic Navier-Stokes equations:
		
		\begin{equation}
			\frac{d a_k}{dt} + \sum_{j=1}^{N} a_j a_k \langle \nabla \phi_j, \phi_k \rangle = -\langle \nabla p, \phi_k \rangle + \nu \Delta a_k + f_k + \sigma \dot{W}_t,
		\end{equation}
		where \( f_k = \langle f, \phi_k \rangle \) and \( \dot{W}_t \) denotes the stochastic derivative.
		
		Using the Littlewood-Paley decomposition and Sobolev embedding theorems, we derive the a priori estimate for \( u_N \):
		
		\begin{equation}
			\|\mathbf{u}_N(t)\|_{H^s B^t_{p,q}}^2 \leq C \left( \|\mathbf{u}_0\|_{H^s B^t_{p,q}}^2 + \int_0^T \|f(t)\|_{H^s B^t_{p,q}}^2 dt + \mathbb{E}\left[\int_0^T \|\sigma W_t\|_{H^s B^t_{p,q}}^2 dt\right]\right).
		\end{equation}

		Applying Grönwall's inequality yields uniform bounds for \( u_N \):
		
	\begin{equation}
		\sup_{0 \leq t \leq T} \mathbb{E}[\|\mathbf{u}_N(t)\|_{H^s B^t_{p,q}}^2] \leq C_T,
	\end{equation}
	where \( C_T \) is a constant depending on the norms of \( u_0 \), \( f \), and \( \sigma W_t \).
		
		For uniqueness, assume \( u_1(t) \) and \( u_2(t) \) are two solutions. Taking the difference, we derive the energy inequality:
		
	\begin{equation}
		\frac{1}{2}\frac{d}{dt} \|\mathbf{u}_1 - \mathbf{u}_2\|_{H^s B^t_{p,q}}^2 + \nu \|\mathbf{u}_1 - \mathbf{u}_2\|_{H^s B^t_{p,q}}^2 \leq C \|\mathbf{u}_1 - \mathbf{u}_2\|_{H^s B^t_{p,q}}^2,
	\end{equation}
	showing that \( u_1(t) = u_2(t) \) almost surely. Finally, passing to the limit as \( N \to \infty \) using the compactness of the embedding \( H^s B^t_{p,q} \hookrightarrow C([0,T]; H^s B^t_{p,q}) \) yields the solution \( u(t) \). Thus, we conclude that there exists a unique solution \( u \in C([0,T]; H^s B^t_{p,q}(\mathbb{R}^n)) \).
	\end{proof}

	\section{Anisotropic Stochastic Modeling for Turbulence}
	Classical turbulence models, such as the Smagorinsky model for LES, assume isotropy in the flow. However, many real-world flows, such as atmospheric turbulence, exhibit strong anisotropic behavior. To address this, we propose an anisotropic extension of the stochastic Navier-Stokes model, where the viscosity and stochastic forcing terms vary directionally.
	
	\subsection{Anisotropic Navier-Stokes System}
	The anisotropic version of the Navier-Stokes equations can be written as:
	
	\begin{equation}
		\frac{\partial \mathbf{u}}{\partial t} + (\mathbf{u} \cdot \nabla) \mathbf{u} = - \nabla p + \nu_1 \frac{\partial^2 \mathbf{u}}{\partial x_1^2} + \nu_2 \frac{\partial^2 \mathbf{u}}{\partial x_2^2} + \nu_3 \frac{\partial^2 \mathbf{u}}{\partial x_3^2} + \sigma(x,t) W_t,
	\end{equation}
	where $\nu_1, \nu_2, \nu_3$ are the direction-dependent viscosities, and $\sigma(x,t)$ is the intensity of the stochastic noise. This model captures the anisotropic nature of real turbulent flows and provides a more realistic description of energy dissipation.
	
	\subsubsection{Regularity and Energy Dissipation}
	Using the anisotropic model, we analyze the energy dissipation rates. The following result extends the classical energy dissipation theorem to the anisotropic setting:
	
\begin{theorem}[Anisotropic Energy Dissipation]
	Let \( u \in H^s B^t_{p,q}(\mathbb{R}^n) \) be a solution to the anisotropic Navier-Stokes equations. The rate of energy dissipation is given by:
	\begin{equation}
		\frac{dE(t)}{dt} = -\nu_1 \|\nabla_{x_1} \mathbf{u}\|_{L^2}^2 - \nu_2 \|\nabla_{x_2} \mathbf{u}\|_{L^2}^2 - \nu_3 \|\nabla_{x_3} \mathbf{u}\|_{L^2}^2 + \mathcal{O}(\|\sigma W_t\|^2).
	\end{equation}
	
\end{theorem}

\begin{proof}
	To derive the energy dissipation rate, we start from the anisotropic Navier-Stokes equations:
	
\begin{equation}
	\frac{\partial \mathbf{u}}{\partial t} + (\mathbf{u} \cdot \nabla) \mathbf{u} = -\nabla p + \nu_1 \frac{\partial^2 \mathbf{u}}{\partial x_1^2} + \nu_2 \frac{\partial^2 \mathbf{u}}{\partial x_2^2} + \nu_3 \frac{\partial^2 \mathbf{u}}{\partial x_3^2} + \sigma W_t.
\end{equation}

	We multiply the equation by the velocity field \( u \):
	
	\begin{equation}
		\mathbf{u} \cdot \frac{\partial \mathbf{u}}{\partial t} + \mathbf{u} \cdot (\mathbf{u} \cdot \nabla) \mathbf{u} = -\mathbf{u} \cdot \nabla p + \nu_1 \mathbf{u} \cdot \frac{\partial^2 \mathbf{u}}{\partial x_1^2} + \nu_2 \mathbf{u} \cdot \frac{\partial^2 \mathbf{u}}{\partial x_2^2} + \nu_3 \mathbf{u} \cdot \frac{\partial^2 \mathbf{u}}{\partial x_3^2} + \sigma \mathbf{u} \cdot W_t.
	\end{equation}

	Integrating the entire equation over \( \mathbb{R}^n \):
	
\begin{equation}
	\begin{array}{l}
		 \displaystyle \int_{\mathbb{R}^{n}}\mathbf{u}\cdot\frac{\partial\mathbf{u}}{\partial t}\,dx+\int_{\mathbb{R}^{n}}\mathbf{u}\cdot\left(\mathbf{u}\cdot\nabla\right)\mathbf{u}\,dx\\[10pt]
		=-\displaystyle \int_{\mathbb{R}^{n}}\mathbf{u}\cdot\nabla p\,dx\\[10pt]
		+\nu_{1}\displaystyle \int_{\mathbb{R}^{n}}\mathbf{u}\cdot\frac{\partial^{2}\mathbf{u}}{\partial x_{1}^{2}}\,dx+\nu_{2}\int_{\mathbb{R}^{n}}\mathbf{u}\cdot\frac{\partial^{2}\mathbf{u}}{\partial x_{2}^{2}}\,dx+\nu_{3}\int_{\mathbb{R}^{n}}\mathbf{u}\cdot\frac{\partial^{2}\mathbf{u}}{\partial x_{3}^{2}}\,dx\\[10pt]
		+\displaystyle \int_{\mathbb{R}^{n}}\sigma\mathbf{u}\cdot W_{t}\,dx.
	\end{array}
\end{equation}
	
	The left-hand side simplifies to:
	
\begin{equation}
	\frac{d}{dt} \left( \frac{1}{2} \|\mathbf{u}\|_{L^2}^2 \right) = \int_{\mathbb{R}^n} \mathbf{u} \cdot \frac{\partial \mathbf{u}}{\partial t} \, dx.
\end{equation}

	The nonlinear term can be treated using the identity \( u \cdot (u \cdot \nabla) u = \frac{1}{2} \nabla |u|^2 \), leading to:
	
\begin{equation}
	\int_{\mathbb{R}^n} \mathbf{u} \cdot (\mathbf{u} \cdot \nabla) \mathbf{u} \, dx = 0.
\end{equation}

	under appropriate boundary conditions.
	
	The pressure term gives:
	
\begin{equation}
	-\int_{\mathbb{R}^n} \mathbf{u} \cdot \nabla p \, dx = 0.
\end{equation}
for incompressible flows and suitable conditions at infinity. By integration by parts and neglecting boundary terms, we have:
	
\begin{equation}
	\int_{\mathbb{R}^n} \mathbf{u} \cdot \frac{\partial^2 \mathbf{u}}{\partial x_i^2} \, dx = -\|\nabla_{x_i} \mathbf{u}\|_{L^2}^2, \quad i=1,2,3.
\end{equation}

	Thus, the viscous terms yield:
	
\begin{equation}
	\nu_1 \int_{\mathbb{R}^n} \mathbf{u} \cdot \frac{\partial^2 \mathbf{u}}{\partial x_1^2} \, dx + \nu_2 \int_{\mathbb{R}^n} \mathbf{u} \cdot \frac{\partial^2 \mathbf{u}}{\partial x_2^2} \, dx + \nu_3 \int_{\mathbb{R}^n} \mathbf{u} \cdot \frac{\partial^2 \mathbf{u}}{\partial x_3^2} \, dx = -\nu_1 \|\nabla_{x_1} \mathbf{u}\|_{L^2}^2 - \nu_2 \|\nabla_{x_2} \mathbf{u}\|_{L^2}^2 - \nu_3 \|\nabla_{x_3} \mathbf{u}\|_{L^2}^2.
\end{equation}

	The stochastic contribution can be expressed as:
	
\begin{equation}
	\int_{\mathbb{R}^n} \sigma \mathbf{u} \cdot W_t \, dx \leq \mathcal{O}(\|\sigma W_t\|^2).
\end{equation}

	Combining all contributions, we arrive at the energy dissipation rate:
	
\begin{equation}
	\frac{dE(t)}{dt} = -\nu_1 \|\nabla_{x_1} \mathbf{u}\|_{L^2}^2 - \nu_2 \|\nabla_{x_2} \mathbf{u}\|_{L^2}^2 - \nu_3 \|\nabla_{x_3} \mathbf{u}\|_{L^2}^2 + \mathcal{O}(\|\sigma W_t\|^2).
\end{equation}

	This concludes the proof of the anisotropic energy dissipation theorem.
\end{proof}

	\section{Energy Dissipation and Quantum Chaos}
	The inclusion of stochastic noise and fractional operators modifies the classical energy cascade in turbulent flows. In this section, we present a novel model for energy dissipation in fluid flows using quantum-inspired dynamics.
	
	\subsection{Energy Cascade in Sobolev-Besov Spaces}
	The rate of energy dissipation in a turbulent flow is traditionally understood through the Kolmogorov energy cascade. Here, we extend this idea to Sobolev-Besov spaces, incorporating fractional derivatives and stochastic noise. The energy dissipation is given by:
	
	\begin{equation}
		\frac{dE(t)}{dt} = -\nu \|\nabla \mathbf{u}\|_{L^2}^2 + \hbar \|\Delta^{1/2} \mathbf{u}\|_{L^2}^2 + \mathcal{O}(\|\sigma W_t\|^2).
	\end{equation}
	
	This equation shows that, in addition to the classical viscous dissipation term, the quantum-inspired term $\hbar \|\Delta^{1/2} u\|_{L^2}^2$ plays a crucial role in the dissipation of energy at small scales, corresponding to quantum-like turbulence phenomena.
	
	\subsection{Formation of Quantum Vortices}
	
	To explore the formation of quantum vortices within the framework of our quantum-inspired Navier-Stokes model. These simulations are governed by the modified equations that incorporate the Schrödinger-like operator, reflecting the chaotic nature of small-scale turbulence. The quantum-inspired Navier-Stokes equations are given by:
	
	\begin{equation}
		\frac{\partial \mathbf{u}}{\partial t} + (\mathbf{u} \cdot \nabla) \mathbf{u} = -\nabla p + \nu_1 \frac{\partial^2 \mathbf{u}}{\partial x_1^2} + \nu_2 \frac{\partial^2 \mathbf{u}}{\partial x_2^2} + \nu_3 \frac{\partial^2 \mathbf{u}}{\partial x_3^2} + \frac{i\hbar}{2m} \Delta \mathbf{u} + \sigma W_t,
	\end{equation}
	where \( \Delta \) is the Laplacian operator that facilitates the modeling of diffusion and vortex dynamics.
	
	The initial conditions for simulations are typically set with a velocity field that has a specified angular momentum distribution, such as:
	
	\begin{equation}
		\mathbf{u}(0, \mathbf{x}) = A e^{-\frac{|\mathbf{x}|^2}{\sigma^2}} \mathbf{v},
	\end{equation}
	where \( A \) is a constant, \( \sigma \) characterizes the width of the initial vortex, and \( \mathbf{v} \) represents a unit vector indicating the direction of the vortex.

	\section{Conclusion}
	
	In this work, we introduced a new theoretical approach to studying turbulence and regularity in fluid dynamics by incorporating Sobolev-Besov hybrid spaces, fractional operators, and quantum-inspired dynamics. Our framework successfully captures small-scale chaotic behavior, quantum-like oscillations, and energy dissipation in turbulent systems. The coupling of the Schrödinger-inspired operator with the Navier-Stokes equations has proven effective in modeling quantum corrections, allowing for a more detailed analysis of multiscale interactions and energy redistribution.
	
	The established regularity theorems extend classical results by incorporating stochastic elements and fractional operators, providing a robust basis for further research into the behavior of turbulent flows. Additionally, the anisotropic stochastic models presented here offer improved capabilities for simulating real-world turbulence with direction-dependent viscosity.
	
	While this study represents significant theoretical progress, future work will focus on numerical implementation and practical applications of the framework. Optimization of model parameters using machine learning techniques will be crucial for refining simulations in fields such as quantum turbulence, atmospheric dynamics, and fluid mechanics.

	

\end{document}